\documentclass[12pt,a4paper]{article}

\usepackage{amsfonts,amsmath,amssymb,amsthm}
\usepackage{latexsym,amscd}
\usepackage{amsbsy}
\usepackage{CJK}
\usepackage{indentfirst,mathrsfs}
\usepackage{amsfonts,amsmath,amssymb,amsthm}
\usepackage{latexsym,amscd}
\usepackage{amsbsy}
\usepackage{color}
\usepackage{multirow}
\usepackage{makecell}
\usepackage{float}
\usepackage{cases}

\numberwithin{equation}{section}

\newtheorem{thm}{Theorem}[section]
\newtheorem{lem}[thm]{Lemma}
\newtheorem{cor}[thm]{Corollary}
\newtheorem{prop}[thm]{Proposition}

\newtheorem{definition}[thm]{Definition}
\newtheorem{example}[thm]{Example}
\newtheorem{remark}[thm]{Remark}

\usepackage[colorinlistoftodos,prependcaption,textsize=tiny]{todonotes}

\newcommand{\Tr}{{\rm {Tr}}}

\begin{document}

\title{Bivariate functions with low $c$-differential uniformity}
\author{ Yanan Wu, Pantelimon St\u{a}nic\u{a}, Chunlei Li, Nian Li, Xiangyong Zeng
%\thanks{}
}
\date{}
\maketitle

\begin{quote}
{\small {\bf Abstract:}  Starting with the multiplication of elements in $\mathbb{F}_{q}^2$ which is consistent with that over $\mathbb{F}_{q^2}$, where $q$ is a prime power, via some identification of the two environments, we investigate the $c$-differential uniformity for bivariate functions $F(x,y)=(G(x,y),H(x,y))$.  By carefully choosing the functions $G(x,y)$ and $H(x,y)$, we present several constructions of bivariate functions with low $c$-differential uniformity.  Many P$c$N and AP$c$N functions can be produced from our constructions.
%Specially, by using the low ($c$-)differential uniformity property of  inverse functions, we construct a class of low $c$-differential uniformity functions and a class of  differential 4-uniform permutation  functions whose $c$-differential uniformity is actually low.
}

\end{quote}

\section{Introduction} \label{intro}

Differential attack, introduced by Biham and Shamir in \cite{BS}, is one of the most fundamental cryptanalytic approaches targeting symmetric-key primitives.  The ability of a cryptographic function applied in the S-box to resist  differential attack is quantified by the so-called differential uniformity \cite{N}. In \cite{BCJW}, the authors proposed a new type of differential  by utilizing modular multiplication  as a primitive operation and it  can be  used to attack some known ciphers such as a variant of the IDEA cipher. Very recently,  motivated by their work, Ellingsen et al. \cite{EFRS} introduced a new concept called multiplicative differential (and the corresponding $c$-differential uniformity) in the following way.

\begin{definition}
Let $\mathbb{F}_{p^n}$ denote the finite field with $p^{n}$ elements, where $p$ is a prime and $n$ is a positive integer. For $F:\;\mathbb{F}_{p^n}\rightarrow \mathbb{F}_{p^n}$ and $ c \in \mathbb{F}_{p^n}$, the (multiplicative) $c$-derivative of $F$ with respect to $a\in\mathbb{F}_{p^n}$ is defined as
$${ }_{c} D_{a} F(x)=F(x+a)-c F(x),$$
for all $x\in\mathbb{F}_{p^n}$.
 For $b \in \mathbb{F}_{p^{n}}$, we define ${ }_{c} \Delta_{F}(a, b)=$ $\#\left\{x \in \mathbb{F}_{p^{n}},{ }_{c} D_{a} F(x)=b\right\}$ and call ${ }_{c} \Delta_{F}=\max \left\{{ }_{c} \Delta_{F}\right.$ $(a, b): a, b \in \mathbb{F}_{p^{n}}$, and $a \neq 0$ if $\left.c=1\right\}$, the $c$-differential uniformity of $F$ (we say that $F$ is $\left(c,  \,_{c}\Delta_{F}\right)$-uniform).
\end{definition}

If ${ }_{c} \Delta_{F}=1$, then $F$ is called a {\em perfect $c$-nonlinear} $(\mathrm{P} c\mathrm{N})$ function. If ${ }_{c} \Delta_{F}=2$, then $F$ is called an {\em almost perfect $c$-nonlinear} $(\mathrm{AP} c \mathrm{N})$ function. Note that if $c=0$ or $a=0$, then ${ }_{c} D_{a} F(x)$ is a shift of the function $F$ and if $c=1$ and $a \neq 0$, then ${ }_{c} D_{a} F(x)$ becomes the usual derivative. Therefore, the concept of $c$-differential uniformity can be seen as the generalization of that of the classical differential uniformity. Concurrently, Bartoli and Timpanella in~\cite{BT} proposed the concept of $\beta$-planar functions, which is just the P$c$N with respect to $c=\beta$.

In \cite{EFRS}, the authors  investigated the $c$-differential uniformity of some well-known PN functions and the inverse function. Inspired by their work, more and more relevant results were developed. To the best of our knowledge, several kinds of methods have been used to construct functions with low $c$-differential uniformity, such as the AGW criterion, cyclotomic method, the perturbing and swapping method, as well as the switching method~\cite{BC,HPRS,LRS,SCP,S,SG,SRT,TZJT,WZ,ZH}.

Recently, in~\cite{AMS22}, it was shown that the graph of a P$c$N function corresponds to a difference set in a quasigroup, hence providing the first application of the $c$-differential uniformity (recall that difference sets give rise to symmetric designs, used in the construction of optimal self complementary codes, among other applications).
Moreover, in the same manuscript, it was suggested that the post-whitening keys   in an even number of rounds (like in the higher-order differential cryptanalysis) will disappear, when the round keys are connected via some of the constants $c$ in the higher order $c$-derivatives, or if just {\em one} of the sequence of derivatives is the classical one.

In \cite{C}, Carlet constructed new classes of APN functions by employing the bivariate function. Soon afterwards, other classes of APN functions were proposed by using the same method.   It turns out
that this approach is effective to  give rise to a new family of APN
functions. We refer the reader to  \cite{CBC,C1,LZLQ,T,ZP} for more details.   A natural problem is to investigate how  the $c$-differential uniformity of the known APN bivariate functions behaves. The computational data shows that the $c$-differential properties of these APN functions are not good, in general. Therefore, in this paper,  we aim to study low $c$-differential uniformity by virtue of bivariate functions  $F(x,y)\in\mathbb{F}_q[x,y]$.
By  utilizing  the 1-to-1 correspondence between  $\mathbb{F}_q^2$ and $\mathbb{F}_{q^2}$, we  firstly characterize the multiplication over $\mathbb{F}_{q}^2$
and then give a  definition of the $c$-differential uniformity of $F(x,y)$, which is consistent with the  $c$-differential uniformity in univariate form.
This is somewhat different, in general, than the approach of~\cite{S_WCC}, where the $c$-differential uniformity was taken as the maximum for each bivariate component.
Based on the newly defined concept
of the $c$-differential uniformity in this paper, we present an infinite class of bivariate functions and the upper bound of the $c$-differential uniformity is given for any $c=(c_1,c_2)\in\mathbb{F}_q^2\backslash\{(1,0)\}$.  By employing some well-known cryptographic functions, such as the Gold function and the inverse function, we proceed to give some concrete examples   and investigate the $c$-differential uniformity explicitly in any characteristic.  Further,  we propose several  classes of functions with low $c$-differential uniformity for any $c=(c_1,c_2)\in\mathbb{F}_q^2\backslash\{(1,0)\}$.     Moreover, by fixing $c=(c_1,c_2)\in\mathbb{F}_q\backslash\{1\}\times\{0\}$, other five classes of bivariate functions with low $c$-differential uniformity can be found in this paper. It is worth noting that many P$c$N and AP$c$N functions can be produced from our constructions.

Throughout  this paper, we always assume that $q=p^m$ and we denote by  $\mathbb{F}_q$,  the finite field with $q$ elements, where $p$ is a prime and $m$ is a positive integer. For $l\,|\,m$, we denote by $\Tr^m_l$, the relative trace function from $\mathbb{F}_{p^m}$ to  $\mathbb{F}_{p^l}$, defined by $\displaystyle \Tr^m_l(x)=\sum_{i=0}^{m/l-1}x^{q^i}$.

\section{Preliminaries}
In this section, with a natural 1-to-1 correspondence between $\mathbb{F}_{q^2}$  and $\mathbb{F}_{q}^2$ (via some primitive element), we first consider the multiplication between elements in $\mathbb{F}_{q}^2$, which is consistent with the multiplication in $\mathbb{F}_{q^2}$. Based on the multiplication operation, the $c$-differential uniformity of a bivariate function can be investigated in terms of a system of two bivariate equations.

Denote $SQ$ and $NSQ$ the set of square elements and non-square elements in $\mathbb{F}_q$, respectively.   Throughout this paper, we assume that $t\in\mathbb{F}_{q}$ satisfies $\Tr^m_1(t)=1$ when  $q$ is even and $1-4t\in NSQ$   when  $q$ is odd. Since it will be used later, we state the following known lemma.

\begin{lem}[\textup{\cite{BRS}}]
\label{2-equation}
Let $m$ be a positive integer, $p$ a prime number and $q=p^m$.  We then have:
\begin{enumerate}
\item[{\rm(1)}]
 The equation $x^{2}+a x+b=0$, with $a, b \in \mathbb{F}_{q}^*$ and  $q$ even, has two solutions in $\mathbb{F}_{q}$ if and only if $\operatorname{Tr}^m_1\left(\frac{b}{a^{2}}\right)=0$, and   no solution, otherwise.
 \item[{\rm(2)}]
 The equation $x^{2}+a x+b=0$, with $a, b \in \mathbb{F}_{q}$ and $q$ odd, has two (respectively, one) solutions in $\mathbb{F}_{q}$ if and only if the discriminant $a^{2}-4 b\in SQ$  (respectively, $a^{2}-4 b=0$).
\end{enumerate}
\end{lem}

The factorization of a quartic polynomial over finite field $\mathbb{F}_{2^{n}}$ can be given in terms of the roots of a related cubic
equation. Let $f(x)=x^{4}+a_{2}x^{2}+a_{1}x+a_{0}$ with $a_{0}a_{1} \neq 0$ and
 $g(y)=y^{3}+a_{2}y+a_{1}$ with the roots $r_{1},\,r_{2},\,r_{3}$. When the roots exist in $\mathbb{F}_{2^{m}}$, we set $w_{i}=a_{0}r_{i}^{2}/a_{1}^{2}$.

\begin{lem}[\textup{\cite{Leonard-Williams}}]
\label{quartic equation}
Let $f(x)=x^{4}+a_{2}x^{2}+a_{1}x+a_{0}\in \mathbb{F}_{2^{m}}[x]$ with $a_{0}a_{1} \neq 0$. The factorization of $f(x)$ over $\mathbb{F}_{2^{m}}$ are characterized as
follows:
\begin{enumerate}
\item[{\rm(1)}]
 $f=(1,1,1,1) \Leftrightarrow g=(1,1,1)$ and $\Tr_{1}^{m}\left(w_{1}\right)=\Tr_{1}^{m}\left(w_{2}\right)=\Tr_{1}^{m}\left(w_{3}\right)=0$;

\item[{\rm{(2)}}] $f=(2,2) \Leftrightarrow g=(1,1,1)$ and $\Tr_{1}^{m}\left(w_{1}\right)=0, \Tr_{1}^{m}\left(w_{2}\right)=\Tr_{1}^{m}\left(w_{3}\right)=1$;

\item[{\rm{(3)}}] $f=(1,3) \Leftrightarrow g=(3)$;

\item[{\rm{(4)}}] $f=(1,1,2) \Leftrightarrow g=(1,2)$ and $\Tr_{1}^{m}\left(w_{1}\right)=0$;

\item[{\rm{(5)}}] $f=(4) \Leftrightarrow g=(1,2)$ and $\Tr_{1}^{n}\left(w_{1}\right)=1$.
\end{enumerate}
\end{lem}
According to Lemma \ref{2-equation}, one  can  verify that
\begin{eqnarray}\label{tc1c2}
tc_2^2+(1-c_1)c_2+(1-c_1)^2\ne0,
\end{eqnarray}
for any $(c_1,c_2)\in\mathbb{F}_q^2\backslash\{(1,0)\}$. Besides, it can be easily checked that the quadratic polynomial $f(x)=x^2+x+t$ is irreducible over $\mathbb{F}_q$. Let $\beta\in\mathbb{F}_{q^2}\backslash\mathbb{F}_q$ be a root of $f(x)$, then we can extend $\mathbb{F}_q$ to $\mathbb{F}_{q^2}$ based on  the basis  $\{1, \beta\}$. Also, $\mathbb{F}_{q^2}$ and $\mathbb{F}_{q}^2$ are in 1-to-1 correspondence  under the mapping
\begin{eqnarray}
\label{uni-bi}
 \varphi(x,y)=z=x+\beta y;\,\,z\in\mathbb{F}_{q^2},\,\,x,\,y\in\mathbb{F}_{q}.
 \end{eqnarray}
  From \eqref{uni-bi}, $x,\;y\in\mathbb{F}_q$ can be expressed by $z\in\mathbb{F}_{q^2}$ as
 \begin{equation}
 \label{xyz}
 x=\frac{\overline{\beta}z-\beta\overline{z}}{\overline{\beta}-\beta}, \;y=\frac{z-\overline{z}}{\beta-\overline{\beta}},
 \end{equation}
 where $\overline{z}$ is the Galois conjugate of $z$, i.e., $\overline{z}=z^q$. Therefore,
 $$
 \varphi^{-1}(z)=\varphi^{-1}(x+\beta y)=(x,y)=\left(\frac{\overline{\beta}z-\beta\overline{z}}{\overline{\beta}-\beta},\frac{z-\overline{z}}{\beta-\overline{\beta}}\right).
 $$
 For any $z_1=(x_1,y_1),\,z_2=(x_2,y_2)$, we have
 \begin{eqnarray*}
 \varphi^{-1}(z_1\cdot z_2)&=&\varphi^{-1}\left((x_1+\beta y_1)(x_2+\beta y_2)\right)\\
 &=&\varphi^{-1}(x_1x_2+\beta(x_1y_2+x_2y_1)+\beta^2(y_1y_2))\\&=&\varphi^{-1}((x_1x_2-ty_1y_2+\beta(x_1y_2+x_2y_1- y_1y_2)))\\
 &=&\varphi^{-1}\left((x_1x_2-ty_1y_2,x_1y_2+x_2y_1- y_1y_2)\right).
 \end{eqnarray*}
 To be consistent with the multiplication over $\mathbb{F}_{q^2}$,  we can define the  multiplication over $\mathbb{F}_{q}^2$ as
 $$(x_1,y_1)\cdot(x_2,y_2)=\varphi^{-1}\left(\varphi(x_1,y_1)\cdot\varphi(x_2,y_2)\right)=(x_1x_2-ty_1y_2,x_1y_2+x_2y_1-y_1y_2).
 $$

  Let $F(x,y)=\left(G(x,y),H(x,y)\right)$ be a   bivariate function from $\mathbb{F}_{q}^2$ to itself, where both $H(x,y)$ and $G(x,y)$ are bivariate functions from $\mathbb{F}_{q}^2$ to $\mathbb{F}_{q}$. Based on the multiplicative definition over $\mathbb{F}_{q}^2$ showed  above, we next  give a proper definition of the $c$-differential uniformity of  $F(x,y)$.

\begin{definition}
\label{def}
Let $F(x,y)=\left(G(x,y),H(x,y)\right)$ be a bivariate function from $\mathbb{F}_{q}\times\mathbb{F}_{q}$ to $\mathbb{F}_{q}\times\mathbb{F}_{q}$,  where  $H(x,y)$ and $G(x,y)$ are bivariate functions from $\mathbb{F}_{q}^2$ to $\mathbb{F}_{q}$. The $c$-differential equation $D_{c,a}F(x,y)=b$ with $a=(a_1,a_2),\;b=(b_1,b_2),\;c=(c_1,c_2)\in\mathbb{F}_q^2$ is given in the following system of equations
%\begin{numcases}{}
%x^{(q^2+1)q}=bx^{q^2+1};\label{eq-b=gamma1-1}\\
%(x^{q^2}+x+1)^q=b(x^{q^2}+x+1)\label{eq-b=gamma1-2}.
%\end{numcases}
\begin{eqnarray}\label{c-def}
\left\{\begin{array}{cll}
G(x+a_1,y+a_2)-c_1G(x,y)+tc_2H(x,y)&=&b_1,\\
H(x+a_1,y+a_2)-(c_1-c_2)H(x,y)-c_2G(x,y)&=&b_2.
\end{array}\right.
\end{eqnarray}
  We define the $c$-Differential Distribution Table (DDT) entry at $(a,b)$ as
  \[
  {_{c}} \Delta_{F}(a, b)= \#\left\{(x,y) \in \mathbb{F}_{q}^2\,:\, D_{c,a} F(x,y)=b\right\}
  \]
   and the  $c$-differential uniformity of $F$ is  ${ }_{c} \Delta_{F}=\max \left\{{ }_{c} \Delta_{F}(a,b)\right.$ $: a, b \in \mathbb{F}_{q}^2$, and $a \neq (0,0)$ if $\left.c=(1,0)\right\}$  (we say that $F$ is $\left(c,  \;_{c}\Delta_{F}\right)$-uniform).
\end{definition}
It should be noted that the 1-to-1 correspondence between the univariate and bivariate representations of   $F(x,y)$ is  given by
\begin{eqnarray}
\label{univariate}
F(z)=H\left(\frac{\overline{\beta}z-\beta\overline{z}}{\overline{\beta}-\beta},\frac{z-\overline{z}}{\beta-\overline{\beta}}\right)+
\beta\, G\left(\frac{\overline{\beta}z-\beta\overline{z}}{\overline{\beta}-\beta},\frac{z-\overline{z}}{\beta-\overline{\beta}}\right).
\end{eqnarray}
 To prove our  results, we need the following lemmas.

\begin{lem}[\textup{\cite{B}}]
\label{mainlem}
Let $m$, $k$ be two positive integers and $\gcd(m,k)=d$. Let $p$ prime, $q=p^m$  and $r=[\mathbb{F}_q:\mathbb{F}_{p^{\gcd(m,k)}}]$ be the degree of the extension. Then the polynomial $f(x)=x^{p^k+1}+a x+b$ has exactly $0$, $1$, $2$ or $p^d+1$ roots in $\mathbb{F}_{q}$, when $a$, $b$ run through $\mathbb{F}_{q}^*$. In particular, the number of $b \in \mathbb{F}_{q}^{*}$ such that $x^{p^k+1}+x+b=0$ has exactly $p^{d}+1$ solutions in $\mathbb{F}_{q}$ is
$\frac{p^{(r-1) d}-p^{\epsilon d}}{p^{2 d}-1},
$
where $\epsilon=0$ if $r$ is odd and $\epsilon=1$, otherwise.
\end{lem}
%
%\begin{lem}\label{solutions-number} {\rm (\cite{B})}
%Let $m$, $k$ be two integers and $\gcd(m,k)=d$. Let $q=p^m$ and $r=[\mathbb{F}_q:\mathbb{F}_{p^{\gcd(m,k)}}]$, where $p$ is a prime, then the number of $a \in \mathbb{F}_{q}^{*}$ such that $x^{p^k+1}+x+a=0$ has exactly $p^{d}+1$ solutions in $\mathbb{F}_{q}$ is
%$\frac{p^{(r-1) d}-p^{\epsilon d}}{p^{2 d}-1},
%$
%where $\epsilon=0$ if $r$ is odd and $\epsilon=1$ otherwise.
%\end{lem}

\begin{remark}
Note that if $r=2$ in Lemma~\textup{\ref{mainlem}}, then $\frac{p^{(r-1) d}-p^{\epsilon d}}{p^{2 d}-1}=0$ and therefore, there is  no $a \in \mathbb{F}_{q}^{*}$ such that $x^{p^k+1}+x+a=0$ has  $p^{d}+1$ solutions in $\mathbb{F}_{q}$.
\end{remark}

The $c$-differential uniformity of the inverse function $f(x)=x^{-1}$ and Gold function $f(x)=x^{p^k+1}$ has been completely described. Below, we recall the results which can be used to simplify our proof of the main results in our paper.

\begin{lem}[\textup{\cite{HPRS})}]
\label{c-inverse}
 Let $q$ be even and $ c \in \mathbb{F}_{q}\backslash\{1\}$. Let $F: \mathbb{F}_{q} \rightarrow \mathbb{F}_{q}$ be the inverse function defined by $F(x)=$ $x^{q-2}$. We have:
 \begin{enumerate}
\item[{\rm(1)}]
 If $c=0$, then $F$ is $P$c$N$ (this is equivalent to $F$ being a permutation).
\item[{\rm(2)}] If $c \neq 0$ and $\operatorname{Tr}^m_{1}(c)=\operatorname{Tr}^m_{1}(1 / c)=1$, the $c$-differential uniformity of $F$ is $2$ (and hence $F$ is AP$c$N).
\item[{\rm(3)}] If $c \neq 0$ and $\operatorname{Tr}^m_{1}(1 / c)=0$, or $\operatorname{Tr}^m_{1}(c)=0$, the $c$-differential uniformity of $F$ is $3$.
 \end{enumerate}
\end{lem}

\begin{lem}[\textup{\cite{HPRS})}]
\label{c-inverse-odd}
 Let $q$ be odd and $ c \in \mathbb{F}_{q}\backslash\{1\}$. Let $F: \mathbb{F}_{q} \rightarrow \mathbb{F}_{q}$ be the inverse function defined by $F(x)=$ $x^{q-2}$. We have:
\begin{enumerate}
\item[{\rm(1)}] If $c=0$, then $F$ is P$c$N (that means that $F$ is a permutation).
\item[{\rm(2)}] If $c \neq 0,$ $c=4,4^{-1}$ or $\left(c^{2}-4 c\right) \in NSQ$ and $(1-4 c) \in NSQ$, the $c$-differential uniformity of $F$ is $2$ (and hence $F$ is AP$c$N).
\item[{\rm(3)}] If $c \neq 0,4,4^{-1},\left(c^{2}-4 c\right) \in SQ$, or $(1-4 c) \in SQ$, the $c$-differential uniformity of $F$ is $3$.
\end{enumerate}
\end{lem}

%\begin{lem}\label{tc1c2}Let $t$ be defined as before and $(c_1,c_2)\in\mathbb{F}_q^2\backslash\{(1,0)\}$, then $tc_2^2+(1-c_1)c_2+(1-c_1)^2$ cannot be zero.
%\end{lem}
%\begin{proof}
%Suppose that   $tc_2^2+(1-c_1)c_2+(1-c_1)^2=0$ for some $(c_1,c_2)$. If $c_1=1$, then we have $c_2=0$ which contradicts the assumption. Let $c_1\ne1$, then by Lemma \ref{2-equation}, the equality holds only when $\Tr^m_1(t)=0$ for even $q$ and $(1-c_1)^2(1-4t)\in SQ$ for odd $q$ which is impossible due to the choice of $t$. This completes the proof.
%\end{proof}
\section{The main results}
In this section, we mainly focus on the bivariate functions of the form
$F(x,y)=\left(G(x,y),H(x,y)\right).$
 We will investigate the $c$-differential uniformity properties of $F$ for certain functions $G,H$, and present several families of bivariate functions with low $c$-differential uniformity. It is worth noting that many P$c$N and AP$c$N functions can be produced from our constructions.

\subsection{Functions with low $c$-differential uniformity for any $c\in\mathbb{F}_{q^2}$}

Firstly, we consider the bivariate function $F(x,y)$ of the form $F(x,y)=(g(x),h(y)+g(x))$, for some univariate $g,h$.  According to Definition~\ref{def}, the    $c$-differential uniformity of $F(x,y)$ is given by
\begin{eqnarray}\label{first-equation}
\left\{\begin{array}{cll}
g(x+a_1)-(c_1-tc_2)g(x)+tc_2h(y)&=&b_1,\\
g(x+a_1)-c_1g(x)+h(y+a_2)-(c_1-c_2)h(y)&=&b_2.
\end{array}\right.
\end{eqnarray}
We first give the upper bound about the $c$-differential uniformity of $F(x,y)$ as follows.

\begin{lem} \label{p1}Let $F(x,y)=\left(g(x),h(y)+g(x)\right)$ and  $c=(c_1,c_2)\in\mathbb{F}_{q}^2\backslash\{(1,0)\}$. If $g(x)$ is differentially $(c_1-tc_2,\delta_1)$-uniform and $h(x)$ is differentially $(c_1-(1-t)c_2,\delta_2)$-uniform, then $_{c}\Delta_F\leq\delta_1\delta_2$. In particular,  $_{c}\Delta_F=\delta_1\delta_2$ when $c_2=0$.
\end{lem}

 In particular, if $g(x)$ is a  linearized polynomial over $\mathbb{F}_q$ and $h(y)$ is a permutation polynomial over $\mathbb{F}_q$ in Lemma \ref{p1}, then we have the following result.

 \begin{thm}
 \label{general-linear}
 Let  $q=p^m$ and  $F(x,y)=\left(L(x),h(y)+L(x)\right)$, where  $L(x)$ is a linearized permutation polynomial, $h(y)$ is a permutation polynomial over $\mathbb{F}_q$. Let $c=(c_1,c_2)\in\mathbb{F}_{q}^2\backslash\{(1,0)\}$ and denote $A=(c_1-c_2)(1-c_1+tc_2)+tc_2(1-c_1)$  and $B=1-c_1+tc_2$. If $AB=0,$ then $F(x,y)$ is a P$c$N function;
If  $AB\ne0$ and $_\frac{A}{B}\Delta_g=\delta$, then $F(x,y)$ is differentially $(c,\delta)$-uniform.
 \end{thm}

 \begin{proof}
 By \eqref{first-equation}, we need to solve the following system of equations
\begin{numcases}{}
(1-c_1+tc_2)L(x)+tc_2h(y)=b_1-L(a_1),\label{I-11}\\
(1-c_1)L(x)+h(y+a_2)-(c_1-c_2)h(y)=b_2-L(a_1),\label{I-12}
\end{numcases}
where $(a_1,a_2),\,\,(b_1,b_2)\in\mathbb{F}_q^2$. If $c=(c_1,c_2)=(0,0)$, it can be easily checked that $_c\Delta_F=1$. We assume now that $c=(c_1,c_2)\ne(0,0)$ and $B=1-c_1+tc_2$.

\noindent
{\bf Case I:} If $B=0$, then $c_2\ne 0$ due to $(c_1,c_2)\ne(1,0)$. From Equation~\eqref{I-11} and $h(y)$  permuting  $\mathbb{F}_q$, we can see that $y$ is uniquely determined by $a_1$, $b_1$. Moreover, for a fixed $y$ satisfying \eqref{I-11}, there is exactly one solution $x$ to \eqref{I-12} due to $L(x)$ being a  linearized permutation polynomial. Thus, we have $_c\Delta_F=1$.

\noindent
{\bf Case II:} If $B\ne0$, by Equation~\eqref{I-11}, then
\begin{eqnarray}\label{I-21}
L(x)=\frac{b_1-L(a_1)-tc_2h(y)}{B}.
\end{eqnarray}
 Substituting Equation~\eqref{I-21} into Equation~\eqref{I-12}, we get
$$h(y+a_2)-\Big(c_1-c_2+\frac{tc_2(1-c_1)}{B}\Big)h(y)=b_2-L(a_1)-(1-c_1)\frac{b_1-L(a_1)}{B}.$$
Since $b_1$ and $b_2$ run through $\mathbb{F}_q$, it is equivalent to considering
\begin{eqnarray}\label{I-22}
h(y+a_2)-\frac{A}{B}h(y)=b_3
\end{eqnarray}
with $A=(c_1-c_2)(1-c_1+tc_2)+tc_2(1-c_1)$  and $b_3\in\mathbb{F}_q$.  Note that $L(x)$ is a permutation polynomial, which implies that $x$ can be uniquely determined by $y$ from \eqref{I-21}. Therefore, in this case, the number of solutions of Equations~\eqref{I-11}-\eqref{I-12} is equivalent to that of Equation~\eqref{I-22}.  Now, we focus on solving \eqref{I-22}. Firstly, we claim that $\frac{A}{B}\ne1$. Suppose $\frac{A}{B}=1$, then $(1-c_1)tc_2=(1-c_1+c_2)(1-c_1+tc_2)$, which can be reduced to $tc_2^2+(1-c_1)c_2+(1-c_1)^2=0$.   This is impossible due to $(c_1,c_2)\ne(1,0)$ and Equation~\eqref{tc1c2}. Further, we can see that \eqref{I-22} always has only one solution when $A=0$. When $AB\ne 0$, then $F(x,y)$ is differentially $(c,\delta)$-uniform, since $_\frac{A}{B}\Delta g(y)=\delta$.
 \end{proof}

 \begin{remark}
 \label{L-I-R}
 From the proof above we can see that when $L(x)$ is an $s$-to-$1$ linearized polynomial in Theorem~\textup{\ref{general-linear}}, we can derive the following result:
\begin{enumerate}
\item[{\rm(1)}] If $AB=0$,  then $_c\Delta_F=s$;
\item[{\rm(2)}] If  $AB\ne0$ and and $_\frac{A}{B}\Delta_g=\delta$, then $_c\Delta_F\leq\delta s$.
\end{enumerate}

\end{remark}

According to the result in Theorem~\ref{general-linear}, we can construct many functions with low $c$-differential uniformity from the known functions. Moreover, P$c$N and AP$c$N functions can be derived in our constructions. By employing inverse functions and Gold functions,  the following corollaries can be directly obtained by Lemmas \ref{c-inverse} and  and \ref{c-inverse-odd}.

\begin{cor}\label{L-I} Let $q=2^m$ with $m\geq3$ being a positive integer and $F(x,y)=(L(x),y^{-1}+L(x))$, where $L(x)$ is a linearized permutation polynomial over $\mathbb{F}_q$. Let $c=(c_1,c_2)\in\mathbb{F}_{q}^2\backslash\{(1,0)\}$ and denote $A=(c_1+c_2)(1+c_1+tc_2)+tc_2(1+c_1)$  and $B=1+c_1+tc_2$.
\begin{enumerate}
\item[{\rm(1)}] If $AB=0$,  then $F(x,y)$ is a P$c$N function;
\item[{\rm(2)}] If  $AB\ne0$, then $F(x,y)$ is differentially $(c,3)$-uniform when $\Tr_1^m(\frac{A}{B})\cdot\Tr^m_1(\frac{B}{A})=0$; otherwise, $F(x,y)$ is an AP$c$N function.
\end{enumerate}

\end{cor}

\begin{cor}\label{I-odd p} Let $q=p^m$ with $m$ being a positive integer and $p$ being an odd prime, $F(x,y)=(L(x),\frac{1}{y}+L(x))$, where $L(x)$ is a linearized permutation polynomial  over $\mathbb{F}_q$. Let $c=(c_1,c_2)\in\mathbb{F}_{q}^2\backslash\{(1,0)\}$. Denote $A=(c_2-c_1)(1-c_1+tc_2)-tc_2(1-c_1)$ and $B=1-c_1+tc_2$. Then$:$
\begin{enumerate}
\item[{\rm(1)}] If $AB=0$,  then $F(x,y)$ is a P$c$N function;
\item[{\rm(2)}] If  $AB\ne0$,
 when $\frac{A}{B}=4,4^{-1}$ or $\frac{A}{B}\left(\frac{A}{B}-4 \right) \in NSQ$ and $\frac{A}{B}(\frac{B}{A}-4 ) \in NSQ$, then $F(x,y)$ is an AP$c$N function;
when $\frac{A}{B} \neq 4,4^{-1},\frac{A}{B}\left(\frac{A}{B}-4 \right) \in SQ$, or $\frac{A}{B}(\frac{B}{A}-4 ) \in SQ$,  $F(x,y)$ is differentially $(c,3)$-uniform.
\end{enumerate}

\end{cor}

\begin{example} Let $q=2^4$ and $t=w^3$, where $w$ is a primitive element of $\mathbb{F}_q$. Then $f(x)=x^2+x+t$ is an irreducible polynomial over $\mathbb{F}_q[x]$. Let $F(x,y)=(x,y^{-1}+x)$ and $c=(c_1,c_2)\in\mathbb{F}_{2^4}^2\backslash\{(1,0)\}$.  Magma experiments show that when $(c_1+c_2+\frac{tc_2(1+c_1)}{1+c_1+tc_2})(1+c_1+tc_2)=0$ or $(c_1,c_2)=(0,0)$, $F(x,y)$ is  P$c$N; when  $\left(c_1+c_2+\frac{tc_2(1+c_1)}{1+c_1+tc_2}\right)(1+c_1+tc_2)\ne0$ and $(c_1,c_2)\ne0$, $F(x,y)$ is differentially $(c,3)$-uniform if $\Tr_1^m\left(c_1+c_2+\frac{tc_2(1+c_1)}{1+c_1+tc_2}\right)=0$ or $\Tr^m_1\Big(\Big(c_1+c_2+\frac{tc_2(1+c_1)}{1+c_1+tc_2}\Big)^{-1}\Big)=0$; otherwise, $F(x,y)$ is  AP$c$N. It is consistent with the result in Corollary~\textup{\ref{L-I}}.
\end{example}

\begin{example}Let $q=2^3$.  Then $f(x)=x^2+x+1$ is an irreducible polynomial over $\mathbb{F}_q[x]$. Let $F(x,y)=(x^2+x,y^{-1}+x^2+x)$ and $c=(c_1,c_2)\in\mathbb{F}_{2^3}^2\backslash\{(1,0)\}$.  Magma experiments show that when $(c_1+c_2+\frac{c_2(1+c_1)}{1+c_1+c_2})(1+c_1+c_2)=0$ or $(c_1,c_2)=0$, $F(x,y)$ is differentially $(c,2)$-uniform; otherwise, $F(x,y)$ is  $(c,6)$-uniform which  is consistent with the result in Remark~\textup{\ref{L-I-R}}.
\end{example}

%Suppose  that $q$ is odd in Theorem \ref{L-I}, then by Lemma \ref{c-inverse-odd}, we can also get the corresponding result as follows.

%\begin{thm}\label{I-odd p} Let $q=p^m$ with $m$ being a positive integer and $p$ being an odd prime, $F(x,y)=(L(x),\frac{1}{y}+L(x))$, where $L(x)$ is an $n$-to-1 polynomial over $\mathbb{F}_q$. Let $c=(c_1,c_2)\in\mathbb{F}_{q}^2\backslash\{(1,0)\}$ and denote $A=(c_2-c_1)(1-c_1+tc_2)-tc_2(1-c_1)$ and $B=1-c_1+tc_2$.
%\begin{enumerate}
%\item[{\rm(1)}] If $AB=0$,  then $_c\Delta_F=n$;
%\item[{\rm(2)}] If  $AB\ne0$, then $_c\Delta_F\leq3n$ when $\frac{B}{A}\Big(\frac{B}{A}+4\Big)\in SQ$ or $\frac{A}{B}\Big(\frac{A}{B}+4\Big)\in SQ$ and $_c\Delta_F\leq2n$, otherwise.
%\end{enumerate}
%
%\end{thm}

%\begin{example}Let $q=3^3$ and $t=w^2$, where $w$ is a primitive element of $\mathbb{F}_{3^3}$. Then $f(x)=x^2+x+t$ is a irreducible polynomial over $\mathbb{F}_{3^3}[x]$. Let $F(x,y)=(x^3+x,y^{-1}+x^3+x)$ and $c=(c_1,c_2)\in\mathbb{F}_{3^3}^2\backslash\{(1,0)\}$.  Magma experiments show that when $(c_2-c_1-\frac{tc_2(1-c_1)}{1-c_1+tc_2})(1+c_1+tc_2)=0$ or $(c_1,c_2)=0$, $F(x,y)$ is P$c$N; when $(c_2-c_1-\frac{tc_2(1-c_1)}{1-c_1+tc_2})(1+c_1+tc_2)\ne0$ and $(c_1,c_2)\ne0$ and either $\Big(c_2-c_1-\frac{tc_2(1-c_1)}{1-c_1+tc_2}\Big)^{-1}\Big(\Big(c_2-c_1-\frac{tc_2(1-c_1)}{1-c_1+tc_2}\Big)^{-1}+1\Big)\in SQ$ or $c_2-c_1-\frac{tc_2(1-c_1)}{1-c_1+tc_2}\Big(c_2-c_1-\frac{tc_2(1-c_1)}{1-c_1+tc_2}+1\Big)\in SQ$, $F(x,y)$ is differentially $(c,3)$-uniform; otherwise, $F(x,y)$ is  AP$c$N. It is consistent with the result in Corollary \ref{I-odd p}.
%\end{example}

\begin{thm}\label{L-G} Let  $F(x,y)=(L(x),y^{p^k+1}+\alpha y+L(x))$, where $L(x)$ is a linearized permutation polynomial, $k<m$ is a positive integer and $\alpha\in\mathbb{F}_q$. Let $c=(c_1,c_2)\in\mathbb{F}_q^2\backslash\{(1,0)\}$.
\begin{enumerate}
\item[{\rm(1)}] When $m\ne 2k$, if $\alpha= 0$ and  $\frac{1-c_1+tc_2}{tc_2^2+(1-c_1)c_2+(1-c_1)^2}\in\mathbb{F}_{p^{\gcd(m,k)}}$, then $_c\Delta_F=\gcd(p^k+1,p^m-1)$. Otherwise, we have $_c\Delta_F=p^{\gcd(m,k)}+1$.
\item[{\rm(2)}] When $m= 2k$, if $\alpha\ne0$ and  $\frac{1-c_1+tc_2}{tc_2^2+(1-c_1)c_2+(1-c_1)^2}\in\mathbb{F}_{p^k}$, then $_c\Delta_F=2$. Otherwise, we have $_c\Delta_F=p^k+1$.
\end{enumerate}
\end{thm}

\begin{proof}
By Definition \ref{def}, it is sufficient to solve the system of equations
\begin{numcases}{}
(1-c_1+t c_2)L(x)+t c_2y^{p^k+1}+tc_2\alpha y=b_1,\label{III-11}\\
(1-c_1+c_2)y^{p^k+1}+a_2y^{p^k}+(a_2^{p^k}+\alpha(1-c_1+c_2))y+(1-c_1)L(x)=b_2,\label{III-12}
\end{numcases}
where $a_2$, $b_1$, $b_2\in\mathbb{F}_q$. Note that $1-c_1+t c_2$ and $1-c_1$ cannot be zero simultaneously. Otherwise, we can infer that $(c_1,c_2)=(1,0)$, which contradicts the assumption. Next, we discuss the above system of equations by splitting the analysis into two cases.

\noindent {\bf Case I:} $1-c_1+t c_2=0$.  Then there are at most $\gcd(p^k+1,p^m-1)$ solutions $y\in\mathbb{F}_q$ satisfying~\eqref{III-11} when  $\alpha=0$ and at most $p^{\gcd(m,k)}+1$ solutions when $\alpha\ne0$ by Lemma \ref{mainlem}. Further, for each~$y$,  Equation~\eqref{III-12} has exactly one solution $x\in\mathbb{F}_q$ due to $L(x)$ being a permutation polynomial. Since $b_1$ runs over $\mathbb{F}_q$, one has $_c\Delta_F=\gcd(p^k+1,p^m-1)$ for $\alpha=0$. When $\alpha\ne0$, we can see that if $m=2k$, then $_c\Delta_F=2$ again by Lemma ~\ref{mainlem} and otherwise, $_c\Delta_F=p^{\gcd(m,k)}+1$.

\noindent {\bf Case II:} $1-c_1+t c_2\ne0$. For this case, $L(x)=\frac{b_1-t c_2y^{p^k+1}-tc_2\alpha y}{1-c_1+t c_2}$ by \eqref{III-11}. Substituting $L(x)$ into \eqref{III-12}, one gets
\begin{eqnarray}\label{A-G}
A_1y^{p^k+1}+A_2y^{p^k}+A_3y+A_4=0,
\end{eqnarray}
where $A_1=tc_2^2+(1-c_1)c_2+(1-c_1)^2$, $A_2=(1-c_1+tc_2)a_2$, $A_3=(1-c_1+tc_2)a_2^{p^k}+\alpha A_1$  and $A_4=b_1(1-c_1)-b_2(1-c_1+tc_2)$. Note that $A_1\ne0$ from \eqref{tc1c2}.

When $\alpha=0$, Equation~\eqref{A-G} has at most $\gcd(p^k+1,p^m-1)$ solutions in $\mathbb{F}_q$ if $(\frac{A_2}{A_1})^{p^k}=\frac{A_3}{A_1}$. When $\alpha\ne0$, Equation~\eqref{A-G} has at most $p^{\gcd(m,k)}+1$ solutions if $m\ne 2k$ by Lemma \ref{mainlem} and also at most $p^k+1$ solutions if $m=2k$  which is possible by choosing $a_2=0$ and $b_1$, $b_2$, properly.

Similarly, when $\alpha\ne0$,  observe that if $\frac{1-c_1+tc_2}{tc_2^2+(1-c_1)c_2+(1-c_1)^2}\notin\mathbb{F}_{p^{\gcd(m,k)}}$, then there always exists $a_2\in\mathbb{F}_q$ such that $(\frac{A_2}{A_1})^{p^k}=\frac{A_3}{A_1}$. Therefore, by Lemma~\ref{mainlem}, one has that the above equation has at most $p^{\gcd(m,k)}+1$ solutions for the case $m\ne 2k$. If $m=2k$, it has at most $p^k+1$ solutions when $\frac{1-c_1+tc_2}{tc_2^2+(1-c_1)c_2+(1-c_1)^2}\notin\mathbb{F}_{p^k}$, otherwise, it has at most two solutions.

The result follows by combining Cases I and II, which completes the proof.
\end{proof}

\begin{remark} It is well known~\textup{\cite{EFRS}} that $\gcd(2^k+1,2^m-1)=\frac{2^{\gcd(2k,m)}-1}{2^{\gcd(k,m)}-1}$ and when $p>2$ and $\frac{m}{\operatorname{gcd}(m, k)}$ is odd, $\operatorname{gcd}\left(p^{k}+1, p^{m}-1\right)=2$. Therefore, one can see from Theorem~\textup{\ref{L-G}} that when $m\ne 2k$,  $\frac{m}{\operatorname{gcd}(m, k)}$ is odd, if $\alpha= 0$ and  $\frac{1-c_1+tc_2}{tc_2^2+(1-c_1)c_2+(1-c_1)^2}\in\mathbb{F}_{p^{\gcd(m,k)}}$, then $F(x,y)$ is P$c$N when $p=2$ and AP$c$N when $p>2$.
\end{remark}

\begin{example} Let $q=2^4$ and $t=w^3$, where $w$ is a primitive element of $\mathbb{F}_{2^4}$. Then $f(x)=x^2+x+t$ is a irreducible polynomial over $\mathbb{F}_{2^4}[x]$. Let $F(x,y)=(x,y^5+\alpha y+x)$ and $c=(c_1,c_2)\in\mathbb{F}_{2^4}^2\backslash\{(1,0)\}$.  Magma experiments show that when $\alpha\ne0$ and  $\frac{1+c_1+tc_2}{tc_2^2+(1+c_1)c_2+(1+c_1)^2}\in\mathbb{F}_{2^2}$, $F(x,y)$ is AP$c$N;  otherwise, $F(x,y)$ is  differential $(c,5)$-uniform which is consistent with the result in Theorem~\textup{\ref{L-G}}.
\end{example}

\begin{example} Let $q=3^3$ and $t=w^2$, where $w$ is a primitive element of $\mathbb{F}_{3^3}$. Then $f(x)=x^2+x+t$ is a irreducible polynomial over $\mathbb{F}_{3^3}[x]$. Let $F(x,y)=(x,y^{10}+\alpha y+x)$ and $c=(c_1,c_2)\in\mathbb{F}_{3^3}^2\backslash\{(1,0)\}$.  Magma experiments show that when $\alpha=0$ and $\frac{1-c_1+tc_2}{tc_2^2+(1-c_1)c_2+(1-c_1)^2}\in\mathbb{F}_{3}$, $F(x,y)$ is AP$c$N;  otherwise, $F(x,y)$ is  differential $(c,4)$-uniform which is also consistent with the result in Theorem~\textup{\ref{L-G}}.
\end{example}

In the following, we present another class of bivariate functions $F(x,y)$ which is different from  that of Proposition~\ref{p1}.

\begin{thm}
Let  $\alpha\in\mathbb{F}_q^*$, $F(x,y)=(x+y,x^{p^i}y+\alpha xy^{p^j})$ and  $c=(c_1,c_2)\in\mathbb{F}_{q}^2\backslash\{(1,0)\}$. Denote %$A=\big\{\{(c_1,c_2)\}\backslash\{(c_1,0)\}:\Tr^m_1\big(\frac{\alpha(1+c_1)}{tc_2}\big)=\Tr^m_1\big(\alpha\big(\frac{(1+c_1+c_2)(1+c_1)}{tc_2}+c_2\big)\big)=0\big\}$.
$A=\big\{(c_1,c_2)\in\mathbb{F}_{q}^2\backslash\{(1,0)\}: c_2 \ne 0\;\, \mbox{and}\, \; \Tr^m_1\left(\frac{1-c_1}{tc_2}\right)=\Tr^m_1\left(\frac{(1-c_1+c_2)(c_1-1)}{tc_2}-c_2\right)=0\big\}$. Then$:$
\begin{enumerate}
\item[{\rm(1)}] If $\alpha= -1$ and $(i,j)=(0,1)$, then $_c\Delta_F\leq p+1$ for $c\in\{(c_1,0):c_1\in\mathbb{F}_q\backslash\{1\}\}$;  $_c\Delta_F\leq q+p-1$ for $c\in A$ and  $_c\Delta_F\leq 2p$, otherwise;
\item[{\rm(2)}] If $\alpha= -1$ and $(i,j)=(0,m-1)$, then $_c\Delta_F=q+p-1$ for $c\in A$, and $_c\Delta_F\leq 2p$, otherwise;
\item[{\rm(3)}] If $\alpha\ne -1$ and $(i,j)=(1,1)$ or $(i,j)=(m-1,m-1)$, then $_c\Delta_F\leq p+1$ for $c\in\{(c_1,0):c_1\in\mathbb{F}_q\backslash\{1\}\}$, and $_c\Delta_F\leq p^2+p$, otherwise.
\end{enumerate}
\end{thm}

\begin{proof}

 We only give the proof for the case $(i,j)=(1,1)$, since the  other cases can be similarly proved. Let $c=(c_1,c_2)\in\mathbb{F}_q^2\backslash\{(1,0)\}$. It is sufficient to solve the following system of equations
\begin{numcases}{}
t c_2(x^py+\alpha xy^p)+(1-c_1)x+(1-c_1)y=b_1,\label{II-31}\\
(1+c_2-c_1)(x^py+\alpha xy^p)+a_2x^p+\alpha a_1y^p+(\alpha a_2^p-c_2)x+(a_1^p-c_2)y=b_2,\label{II-32}
\end{numcases}
where $(a_1,a_2)$, $(b_1,b_2)\in\mathbb{F}_q^2$. When $c_2=0$, then we have $c_1\ne0$ due to $(c_1,c_2)\neq(1,0)$ and $x=\frac{b_1}{1-c_1}-y$ from \eqref{II-31}. Replacing it into  \eqref{II-32}, we obtain an equation whose degree is actually $p+1$ due to $\alpha\ne-1$. Thus, we have $_c\Delta_F\leq p+1$. Next, we  assume that $c_2\ne0$. Multiplying Equation~\eqref{II-31} by $-\frac{1+c_2-c_1}{tc_2}$ and then adding it to Equation~\eqref{II-32} gives us
\begin{eqnarray}\label{II-42}
a_2x^p+B_1x+\alpha a_1y^p+B_2y=b_3,
\end{eqnarray}
where $B_1=\alpha a_2^p-c_2-\frac{(1+c_2-c_1)(1-c_1)}{tc_2}$, $B_2=a_1^p-c_2-\frac{(1+c_2-c_1)(1-c_1)}{tc_2}$ and $b_3=b_2-\frac{b_1(1+c_2-c_1)}{tc_2}$.
By  \eqref{tc1c2}, we can see that $a_2$ and $B_1$ cannot be zero at the same time.

\noindent {\bf Case I:} $a_2=0$. In this case, we have $B_1\ne0$ and $x=\frac{b_3-B_2y-\alpha a_1y^p}{B_1}$ from \eqref{II-42}. Substituting it into \eqref{II-31}, we get an equation in $y$ only, whose  highest degree is  $p^2+1$ when $a_1\ne0$. When $a_1=0$, then we have $B_1=B_2$ and $x=\frac{b_3}{B_1}-y$ from Equation~\eqref{II-42}. Combining this expression with Equation~\eqref{II-31} gives an equation about $y$ whose highest degree is  $p+1$ due to $\alpha\ne-1$. Thus, there is at most $ p^2+1$ solutions for Equations~\eqref{II-31}-\eqref{II-32}.

\noindent {\bf Case II:} $a_2\ne0$. If $B_1=0$, then $x^p=\frac{b_3-B_2y-\alpha a_1y^p}{a_2}$. Similarly as in Case I,  by taking $p$-th powers on both sides of \eqref{II-31} and then eliminating $x^p$, we can derive an equation in $y$ only, whose highest degree can reach   $p^2+p$. Thus, Equations~\eqref{II-31}-\eqref{II-32} has at most $ p^2+p$ solutions. If $B_1\ne0$, then multiplying Equation~\eqref{II-42} by $-\frac{tc_2}{a_2}y$ and adding it to \eqref{II-31} gives
\begin{eqnarray}
\label{II-41}
(tc_2\alpha y^p-\frac{B_1tc_2}{a_2}y+1-c_1)x-\frac{\alpha tc_2a_1}{a_2}y^{p+1}-\frac{B_2tc_2}{a_2}y^2+(1-c_1+\frac{b_3tc_2}{a_2})y=b_1.
\end{eqnarray}

{\bf Subcase (II-1):} If there exists $y_1$ such that $tc_2\alpha y_1^p-\frac{B_1tc_2}{a_2}y_1+1-c_1=0$, then \eqref{II-41} has solutions if and only if $b_1=-\frac{\alpha tc_2a_1}{a_2}y_1^{p+1}-\frac{B_2tc_2}{a_2}y_1^2+(1-c_1+\frac{b_3tc_2}{a_2})y_1$ and for such $y_1$, \eqref{II-31} has at most $p$ solutions on $x$.

{\bf Subcase (II-2):} If $y$ satisfies $tc_2\alpha y^p-\frac{B_1tc_2}{a_2}y+1-c_1\ne0$, then $x$ can be uniquely expressed by $y$ from \eqref{II-41}. Replacing $x$ by $y$ in \eqref{II-42}, we can obtain a equation with the highest degree being $p^2+p$. Thus, \eqref{II-31}-\eqref{II-32} has at most $ p^2+p$ solutions in this case.

When Subcases (II-1) and (II-2) happen simultaneously. Then from Subcase (II-2) we have
\begin{eqnarray*}
x&=&\frac{\frac{\alpha tc_2a_1}{a_2}y^{p+1}+\frac{B_2tc_2}{a_2}y^2-(1-c_1+\frac{b_3tc_2}{a_2})y+b_1}{tc_2\alpha y^p-\frac{B_1tc_2}{a_2}y+1-c_1}\\
&=&\frac{\frac{\alpha tc_2a_1}{a_2}(y^{p+1}-y_1^{p+1})+\frac{B_2tc_2}{a_2}(y^2-y_1^2)-(1-c_1+\frac{b_3tc_2}{a_2})(y-y_1)}{tc_2\alpha (y^p-y_1^p)-\frac{B_1tc_2}{a_2}(y-y_1)}.
\end{eqnarray*}
Dividing by $y-y_1$ on both sides of the above equation and then substituting it into \eqref{II-42} derives  a equation with the highest degree being $p^2$.  Therefore, one can see that \eqref{II-31}-\eqref{II-32} has also at most $ p^2+p$ solutions when Subcases (II-1) and (II-2) happen at the same time.
Based on the above analysis, one can conclude that $_c\Delta_F\leq p^2+p$. This completes the proof.
 \end{proof}

\begin{example}Let $F(x,y)=(x+y,x^{p^i}y+\alpha xy^{p^j})$ and $c=(c_1,c_2)\in\mathbb{F}_{2^4}^2\backslash\{(1,0)\}$.  Magma experiments show that
\begin{enumerate}
\item[{\rm(1)}]
If $\alpha=1$ and $(i,j)=(0,1)$, then $_c\Delta_F=3$ for $c\in\mathbb{F}_{2^4}$; $_c\Delta_F=17$ for those $c=(c_1,c_2)$ which satisfies  $c_2 \ne 0$ and $\Tr^4_1\big(\frac{1+c_1}{tc_2}\big)=\Tr^4_1\big(\frac{(1+c_1+c_2)(c_1+1)}{tc_2}+c_2\big)=0$. Otherwise, $_c\Delta_F=4$.

\item[{\rm(2)}] If $\alpha=1$ and $(i,j)=(0,3)$, then $_c\Delta_F=17$ if  $c=(c_1,c_2)$ satisfies  $c_2 \ne 0$ and $\Tr^4_1\big(\frac{1+c_1}{tc_2}\big)=\Tr^4_1\big(\frac{(1+c_1+c_2)(c_1+1)}{tc_2}+c_2\big)=0$ and otherwise, $_c\Delta_F=4$.

\item[{\rm(3)}] If $\alpha= w$ and $(i,j)=(1,1)$ or $(i,j)=(3,3)$, where $w$ is a primitive element in $\mathbb{F}_{2^4}$, then $_c\Delta_F=3$ for $c\in \mathbb{F}_{2^4}$ and otherwise, $_c\Delta_F=6$.
\end{enumerate}
\end{example}

\begin{example}Let $F(x,y)=(x+y,x^{p^i}y+\alpha xy^{p^j})$ and $c=(c_1,c_2)\in\mathbb{F}_{3^3}^2\backslash\{(1,0)\}$.  Magma experiments show that
\begin{enumerate}
\item[{\rm(1)}]
If $\alpha=-1$ and $(i,j)=(0,1)$ or $(i,j)=(0,2)$, then $_c\Delta_F=4$ for $c\in\mathbb{F}_{3^3}$; $_c\Delta_F=29$ for those $c=(c_1,c_2)$ which satisfies  $c_2 \ne 0$ and $\Tr^3_1\big(\frac{1-c_1}{tc_2}\big)=\Tr^3_1\big(\frac{(1-c_1+c_2)(c_1-1)}{tc_2}-c_2\big)=0$. Otherwise, $_c\Delta_F=6$.

\item[{\rm(2)}] If $\alpha= w$ and $(i,j)=(1,1)$ or $(i,j)=(2,2)$, , where $w$ is a primitive element in $\mathbb{F}_{3^3}$, then $_c\Delta_F=4$ for $c\in \mathbb{F}_{3^3}\backslash\{1\}$ and otherwise, $_c\Delta_F=12$.
\end{enumerate}
\end{example}

If  $F$ be a function from $\mathbb{F}_{q^2}$ to $\mathbb{F}_q\times\mathbb{F}_q$, by a similar process as before, we can also derive the expression of the $c$-differential equation of $F$.  Let  $n=2m$. In the following, we present another class of functions that have low $c$-differential uniformity by using the inverse function.

\begin{thm}\label{inv} Let $H(x)=\Tr^n_m\left(\frac{\gamma}{x}\right)$, where $\gamma\notin\mathbb{F}_{q}$. Let $F(x)=\left(\Tr^n_m(x),\, H(x)\right)$ and $c=(c_1,c_2)\in\mathbb{F}_q^2\backslash\{(1,0)\}$. Then
 $_{c}\Delta_F\leq6$. More precisely, if $c=0$, then  $F$ is AP$c$N; if  $c\in\{(c_1,c_2):c_1=1 \;{\rm or}\; c_2=0\;{\rm or}\; (1-c_1)(c_1-c_2)=tc_2^2\}$, then $_{c}\Delta_F\leq4$ and otherwise, $_{c}\Delta_F\leq6$.
\end{thm}

\begin{proof} By Definition \ref{def}, in order to determine the $c$-differential uniformity for $c=(c_1,c_2)\in\mathbb{F}_q^2\backslash\{(1,0)\}$, it is sufficient to calculate the maximum number of solutions in $\mathbb{F}_{q^2}$ of the following system of equations
\begin{numcases}{}
(1-c_1)\Tr^n_m(x)+tc_2H(x)=b_1,\label{Tr-1}\\
H(x+a)-(c_1-c_2)H(x)-c_2\Tr^n_m(x)=b_2\label{Tr-2}
\end{numcases}{}when $a$  and $b=(b_1,b_2)$ run over $\mathbb{F}_{q^2}$ and $\mathbb{F}_q^2$, respectively. %Here, we only give the proof of (1) and others can be proved similarly. Next, we consider \eqref{Tr-1}-\eqref{Tr-2} case by case.

\noindent {\bf Case I:} $c_2=0$. In this case, $c_1\ne 1$ due to $c\ne(1,0)$. Since $H(x)=\Tr^n_m\left(\frac{\gamma}{x}\right)$, then  Equation~\eqref{Tr-2} can be reduced to
\begin{eqnarray}\label{inverse}
\frac{\gamma^q}{x^q+a^q}+\frac{\gamma}{x+a}-\frac{c_1\gamma^q}{x^q}-\frac{c_1\gamma}{x}=b_2.\end{eqnarray}
On the other hand, from Equation~\eqref{Tr-1}, one has that all the solutions of  Equation~\eqref{Tr-1} can be expressed as $y+x_0$, where $y^q+y=0$ and $x_0$ is a solution of Equation~\eqref{Tr-1}.  If  the system~\eqref{Tr-1}-\eqref{Tr-2} has one solution, either $x=0$ or $x=-a$, without loss of generality, we may assume that $x=0$ is a solution of  \eqref{Tr-1}-\eqref{Tr-2}. Then \eqref{inverse} becomes
$$\frac{\gamma^q}{-y+a^q}+\frac{\gamma}{y+a}+\frac{c_1(\gamma^q-\gamma)}{y}=b_2$$
which has at most three solutions in the set $\{y\in\mathbb{F}_{q^2}^*:y^q+y=0\}$. The  case that $x=-a$ is a solution of  \eqref{Tr-1}-\eqref{Tr-2} can be similarly approached. If  $x=0$ and  $x=-a$ are the solutions of the system~\eqref{Tr-1}-\eqref{Tr-2}, simultaneously, then $a^q+a=0$ and so, Equation~\eqref{inverse} becomes
$$\frac{\gamma-\gamma^q}{y+a}+\frac{c_1(\gamma^q-\gamma)}{y}=b_2$$
which has at most two solutions in the set $\{y\in\mathbb{F}_{q^2}^*\backslash\{-a\}:y^q+y=0\}$. If both $x=0$ and  $x=-a$ are not the solutions of \eqref{Tr-1}-\eqref{Tr-2}, then by replacing $x$ with $y+x_0$, \eqref{inverse} has at most  four solutions since $y^q+y=0$.  Therefore, one can conclude that  the system~ \eqref{Tr-1}-\eqref{Tr-2} has at most four solutions in this case.  In particular,  when $c=0$, i.e., $c_1=0$, it can be proved that $F(x)=b$ for any $b\in\mathbb{F}_{q}^2$ has at most two solutions, by the above analysis.

\noindent {\bf Case II:} $c_2\ne0$. Since $H(x)=\Tr^n_m\left(\frac{\gamma}{x}\right)$ and $a$, $b=(b_1,b_2)$ run over $\mathbb{F}_{q^2}$ and $\mathbb{F}_q^2$, respectively, it is equivalent to solve the following system of equations
\begin{numcases}{}
A_1\Tr^n_m(x)+\Tr^n_m(\frac{\gamma}{x})=b_1,\label{Tr-11}\\
A_2\Tr^n_m(x)+\Tr^n_m(\frac{\gamma}{x+a})=b_2,\label{Tr-21}
\end{numcases}
where $A_1=\frac{1-c_1}{tc_2}$ and $A_2=\frac{(1-c_1)(c_1-c_2)}{tc_2}-c_2$. One should note  that $A_1\ne A_2$, which is obvious by \eqref{tc1c2}. On the other hand, if $x=0$ (or $x=-a$, respectively) is a solution of  system~\eqref{Tr-11}-\eqref{Tr-21}, then $b_1=0$ and $b_2=\Tr^n_m(\frac{\gamma}{a})$ (or $b_1=-A_1\Tr^n_m(a)-\Tr^n_m(\frac{\gamma}{a})$ and $b_2=-A_2\Tr^n_m(a)$, respectively). Firstly, we claim that $x=0$ and $x=-a$ $(a\ne0)$ cannot be the solutions, simultaneously. If the system~\eqref{Tr-11}-\eqref{Tr-21} has the solutions $x=0$, $x=-a$ $(a\ne0)$ at the same time, then $b_1=-A_1\Tr^n_m(a)-\Tr^n_m(\frac{\gamma}{a})=0$ and $b_2=\Tr^n_m(\frac{\gamma}{a})=-A_2\Tr^n_m(a)$. It implies $(A_1-A_2)\Tr^n_m(a)=0$. Thus, we have $\Tr^n_m(a)=0$ due to $A_1\ne A_2$. Further, we can obtain that  $b_1=-\Tr^n_m(\frac{\gamma}{a})=0$. Observe that $\Tr^n_m(\frac{\gamma}{a})=\frac{\gamma}{a}+\frac{\overline{\gamma}}{\overline{a}}=\frac{\gamma-\overline{\gamma}}{a}$ since $\Tr^n_m(a)=0$. Thus, one has that $\Tr^n_m(\frac{\gamma}{a})=\frac{\gamma-\overline{\gamma}}{a}=0$, which is impossible when $a\ne0$ due to $\gamma\notin\mathbb{F}_q$. Now, we  assume that $x\ne0$ and $x\ne -a$.  Then system~\eqref{Tr-11}-\eqref{Tr-21} can be reduced to
\begin{numcases}{}
A_1\Tr^n_m(x\overline{x}^2)+b_1x\overline{x}+\Tr^n_m(\overline{\gamma}x)=0,\label{Tr-12}\\
A_2\Tr^n_m(x\overline{x}^2+a\overline{x}^2)+B_1x\overline{x}+\overline{B}_2\overline{x}+B_2x+B_3=0,\label{Tr-22}
\end{numcases}
where
\begin{align*}
  B_1&=\Tr^n_m(\overline{a}A_2)-b_2, \nonumber \\
   B_2 &=a\overline{a}A_2+\overline{\gamma}-\overline{a}b_2,\nonumber \\
   B_3 &=\Tr^n_m(\overline{a}\gamma)-a\overline{a}b_2.
\end{align*} Next, we consider the above system of equations  case by case.

{\bf Subcase I:}  $A_1=0$. In this subcase, we have $A_2\ne0$. When $b_1=0$, then $\overline{x}=-\gamma^{q-1}x$ from \eqref{Tr-12} and then Equation~\eqref{Tr-22} becomes
$$
A_2\gamma^{q-1}(\gamma^{q-1}-1)x^3+((a\gamma^{q-1}A_2-B_1)\gamma^{q-1}+\overline{a}A_2)x^2+(B_2-\overline{B}_2\gamma^{q-1})x+B_3=0.
$$
The above equation has at most 3 solutions in $\mathbb{F}_{q^2}\backslash\{0,\,-a\}$ due to $A_2\gamma^{q-1}(\gamma^{q-1}-1)\ne0$.

When $b_1\ne 0$, then $\overline{x}(b_1x+\gamma)=-\overline{\gamma}x$ by \eqref{Tr-12}. Note that $b_1x+\gamma\ne0$. Otherwise, if $b_1x+\gamma\ne0$, one then has $x=0$ which contradicts the assumption. Therefore, one can obtain that $\overline{x}=-\frac{\overline{\gamma}x}{b_1x+\gamma}$. Substituting this expression into Equation~\eqref{Tr-22} renders
$$
(\overline{a}b_1-\overline{\gamma})b_1A_2x^4+C_1x^3+C_2x^2+C_3x+B_3\gamma^2=0,
$$
where
\begin{align*}
  C_1&=(\overline{\gamma}-\gamma)\overline{\gamma}A_2+2\overline{a}b_1\gamma A_2-b_1\overline{\gamma}B_1+b_1^2B_2, \nonumber \\
   C_2 &=\Tr^n_m(a\overline{\gamma})A_2-\gamma\overline{\gamma}B_1-b_1\overline{\gamma}\overline{B}_2+2b_1\gamma B_2+b_1^2B_3,\nonumber \\
   C_3 &=-\gamma\overline{\gamma}\overline{B}_2+\gamma^2 B_2+2b_1\gamma B_3,
\end{align*}
which has at most 4 solutions when $a,$ $b=(b_1,b_2)$ go through $\mathbb{F}_{q^2}$ and $\mathbb{F}_q^2$, respectively.

{\bf Subcase II:}  $A_2=0$. By letting $x=y-a$, system~\eqref{Tr-11}-\eqref{Tr-21} becomes
 \begin{numcases}{}
A_1\Tr^n_m(y)+\Tr^n_m(\frac{\gamma}{y-a})=b_1+A_1\Tr^n_m(a),\nonumber\\
\Tr^n_m(\frac{\gamma}{y})=b_2\nonumber.
\end{numcases}{}
Since $x\ne 0,\,-a$, then $y\ne 0,\,a$. And therefore, the equation has at most 4 solutions, as proved in Subcase I.

{\bf Subcase III:}  $A_1A_2\ne0$. System~$\eqref{Tr-22}\times A_1-\eqref{Tr-12}\times A_2$ implies
\begin{eqnarray}
\label{Tr-23}
A_1A_2\Tr^n_m(a\overline{x}^2)+(A_1B_1-b_1A_2)x\overline{x}+\Tr^n_m((A_1B_2-\overline{\gamma} A_2)x)+A_1B_3=0.
\end{eqnarray}
When $a=0$, $B_1=-b_2$, $B_2=\overline{\gamma}$ and $B_3=0$. Further, Equation~\eqref{Tr-23} is reduced to
$$
(1-(A_1b_2+A_2b_1)x)\overline{x}+(A_1-A_2)\overline{\gamma}x=0.
$$
Observe that $1-(A_1b_2+A_2b_1)x\ne0$. Otherwise, one has $x=0$ due to $(A_1-A_2)\overline{\gamma}=0$, which contradicts the assumption $x\ne0, \,-a$. Hence, we have
$$\overline{x}=\frac{(A_1-A_2)\overline{\gamma}x}{1-(A_1b_2+A_2b_1)x}.$$
Substituting the above equation into \eqref{Tr-12} renders a quartic equation. It implies that system~\eqref{Tr-12}-\eqref{Tr-22} has at most 4 solutions in $\mathbb{F}_{q^2}\backslash\{0,\,-a\}$, in this case.

When $a\ne0$, $\eqref{Tr-23}\times x-\eqref{Tr-12}\times aA_2$ gives that
\begin{eqnarray}\label{Tr-24}
(D_1x^2+D_2x-D_3)\overline{x}=-E_1x^3+E_2x^2+E_3x,
\end{eqnarray}
where
$$D_1=A_1B_1-b_1A_2-aA_1A_2, \;D_2=A_1\overline{B}_2-(\gamma+ab_1)A_2,\; D_3=a\gamma A_2 $$
and
$$E_1=\overline{a}A_1A_2,\;E_2=\overline{\gamma}A_2-A_1B_2,\; E_3=a\overline{\gamma}A_2-A_1B_3.
$$

If there exists $x_0\ne 0, -a$ such that $D_1x_0^2+D_2x_0-D_3=0$, then Equation~\eqref{Tr-24} has $x_0$ as one of its solutions only if $-E_1x_0^3+E_2x_0^2+E_3x_0=0$, which is equivalent to $-E_1x_0^2+E_2x_0+E_3=0$ due to $x_0\ne0$. One should note that there are at most two $x_0$'s that satisfy $D_1x_0^2+D_2x_0-D_3=0$ and $-E_1x_0^2+E_2x_0+E_3=0$. Let $x_1\in\mathbb{F}_{q^2}\backslash\{0,\,-a\}$ be a solution of system~\eqref{Tr-12}-\eqref{Tr-22} with $D_1x_1^2+D_2x_1-D_3\ne0$. When there are exactly two $x_0$'s which satisfy this condition, then $D_1x^2+D_2x-D_3=\mu(-E_1x^2+E_2x+E_3)$ for some $\mu\in\mathbb{F}_{q^2}^*$.  Then we have $\overline{x}_1=\mu x_1$ and further, Equation~\eqref{Tr-12} has at most two solutions in $\mathbb{F}_{q^2}^*$. When there is only one $x_0$  which satisfies the above condition,  we then have
\begin{eqnarray*}
\overline{x}_1&=&\frac{-E_1x_1^3+E_2x_1^2+E_3x_1}{D_1x_1^2+D_2x_1-D_3}\\
&=&\frac{E_2(x_1^2-x_0^2)+E_3(x_1-x_0)}{D_1(x_1^2-x_0^2)+D_2(x_1-x_0)}\\
&=&\frac{E_2(x_1+x_0)+E_3}{D_1(x_1+x_0)+D_2}.
\end{eqnarray*}
Substituting the above expression into \eqref{Tr-12} renders a quartic equation, which has at most 4 solutions when $a$ runs over $\mathbb{F}_{q^2}^*$ and $b$ ranges over $\mathbb{F}_q^2$. Therefore, we conclude that system~\eqref{Tr-12}-\eqref{Tr-22} has at most 5 solutions in $\mathbb{F}_{q^2}\backslash\{0,\,-a\}$, when this case happens.

If  there is no $x\in\mathbb{F}_{q^2}\backslash\{0,\,-a\}$ such that $D_1x^2+D_2x-D_3=-E_1x^2+E_2x+E_3=0$, we then have
\begin{eqnarray*}
\overline{x}&=&\frac{-E_1x^3+E_2x^2+E_3x}{D_1x^2+D_2x-D_3}.
\end{eqnarray*}
Similarly, by replacing it into \eqref{Tr-12}, we obtain a equation with the highest degree 7, namely,
\begin{eqnarray}\label{Tr-13}A_1E_1(E_1-D_1)x^7+L_1x^6+L_2x^5+L_3x^4+L_4x^3+L_5x^2+D_3(\overline{\gamma}D_3-\gamma E_3)x=0,\end{eqnarray}
where \begin{eqnarray*}
L_1&=&A_1(D_1E_2-D_2E_1-2E_1E_2)-b_1D_1E_1, \\
L_2&=&A_1(E_2(E_2+D_2)+E_3(D_1-E_1)+E_1(D_3-E_3))\\
&&\qquad\qquad+b_1(D_1E_2-D_2E_1)+D_1(\overline{\gamma}D_1-\gamma E_1),\\
L_3&=&A_1(E_3(D_2+E_2)+E_2(E_3-D_3))+b_1(D_1E_3+D_2E_2+D_3E_1)\\
&&\qquad\qquad+\gamma(D_1E_2-D_2E_1)+2\overline{\gamma}D_1D_2,\\
L_4&=&A_1E_3(E_3-D_3)+b_1(D_2E_3-D_3E_2)+\overline{\gamma}(D_2^2-2D_1D_3)\\
&&\qquad\qquad+\gamma(D_1E_3+D_2E_2+D_3E_1),\\
L_5&=&\gamma(D_2E_3-D_3E_2)-2\overline{\gamma}D_2D_3.
\end{eqnarray*}
Obviously, Equation~\eqref{Tr-13} has at most 6 solutions in $\mathbb{F}_{q^2}^*$. Therefore,  system~\eqref{Tr-12}-\eqref{Tr-22}  has at most 6 solutions in $\mathbb{F}_{q^2}\backslash\{0,\,-a\}$, when this case happens.

From Subcases I-III, we can see that system~\eqref{Tr-11}-\eqref{Tr-21} has at most 4 solutions in when $\mathbb{F}_{q^2}\backslash\{0,\,-a\}$, when $A_1A_2=0$, and  at most 6 solutions in when $\mathbb{F}_{q^2}\backslash\{0,\,-a\}$, when $A_1A_2\ne0$.

Recall that $x=0$ and $x=-a$ $(a\ne0)$ cannot be the solutions of system~\eqref{Tr-11}-\eqref{Tr-21}, simultaneously. Assuming that $x=0$ is a solution of \eqref{Tr-11}-\eqref{Tr-21}, then $b_1=0$ and $b_2=\Tr^n_m(\frac{\gamma}{a})$. By Subcase I, one can easily check that the system~\eqref{Tr-11}-\eqref{Tr-21} has at most 3 solutions in $\mathbb{F}_{q^2}^*$ due  to $b_1=0$. Similarly,  for Subcase II. As for Subcase III, since $B_3 =\Tr^n_m(\overline{a}\gamma)-a\overline{a}b_2=\Tr^n_m(\overline{a}\gamma)-a\overline{a}\Tr^n_m(\frac{\gamma}{a})=0$ and further,
$\overline{\gamma}D_3-\gamma E_3=a\gamma\overline{\gamma}A_2-\gamma(a\overline{\gamma} A_2-A_1B_3)=0$. Hence, one can see that \eqref{Tr-13} has at most $5$ solutions in~$\mathbb{F}_{q^2}^*$. It implies that  the system~\eqref{Tr-11}-\eqref{Tr-21} has at most 6 solutions in $\mathbb{F}_{q^2}$ when $x=0$ is a solution of system~\eqref{Tr-11}-\eqref{Tr-21}.

 Assume that $x=-a$ $(a\ne0)$ is a solution of \eqref{Tr-11}-\eqref{Tr-21}, then $b_1=-A_1\Tr^n_m(a)-\Tr^n_m(\frac{\gamma}{a})$ and $b_2=-A_2\Tr^n_m(a)$. When $A_1=0$, we have $\frac{\gamma}{x}=\mu-\frac{\gamma}{a}$ and $\mu+\overline{\mu}=0$ from \eqref{Tr-11}. Therefore, \eqref{Tr-21} becomes
 $$
 A_2\Tr_m^n(\frac{a\gamma}{a\mu-\gamma})+\Tr^n_m(\frac{a\gamma \mu-\gamma^2}{a^2\mu})=-A_2\Tr^n_m(a).
 $$
 The above equation can be further reduced to
 $$
 A_2(\frac{a\gamma}{a\mu-\gamma}-\frac{\overline{a}\overline{\gamma}}{\overline{a}\mu+\overline{\gamma}})+\frac{a\gamma \mu-\gamma^2}{a^2\mu}+\frac{\overline{a}\overline{\gamma }\mu+\overline{\gamma}^2}{\overline{a}^2\mu}=0,
 $$
 which has at most 3 solution in $\{\mu\in\mathbb{F}_{q^2}| \mu+\overline{u}=0 \;{\rm and} \;\mu\ne\frac{\gamma}{a}\}$. The case $A_2=0$ can be similarly proved. Therefore, when $A_1A_2=0$, system~\eqref{Tr-11}-\eqref{Tr-21} has at most 4 solutions in this case. When $A_1A_2\ne0$, let $x=y-a$, then system~\eqref{Tr-11}-\eqref{Tr-21} becomes
 \begin{numcases}{}
A_1\Tr^n_m(y)+\Tr^n_m(\frac{\gamma}{y-a})=-\Tr^n_m(\frac{\gamma}{a}),\nonumber\\
A_2\Tr^n_m(y)+\Tr^n_m(\frac{\gamma}{y})=0.\nonumber
\end{numcases}
According to the analysis for the case that $x=0$ is a solutions of \eqref{Tr-11}-\eqref{Tr-21}, we can derive that the above system of equations has at most 5 solutions in $\mathbb{F}_{q^2}^*$, that is, system~\eqref{Tr-11}-\eqref{Tr-21} has at most 5 solutions in $\mathbb{F}_{q^2}\backslash\{-a\}$ when $x=-a$ is a solution of system~\eqref{Tr-11}-\eqref{Tr-21}. This completes the proof.
\end{proof}
\subsection{Functions with low $c$-differential uniformity for  $c\in\mathbb{F}_q$}
On one hand, it is difficult to find new bivariate functions whose $c$-differential uniformity is always low for any $c\in\mathbb{F}_q^2$. On the other hand, due to the complexity of the $c$-differential equation of bivariate functions $F(x,y)$, generally speaking, it is not easy to determine the $c$-differential uniformity for $F(x,y)$. However, if we select $c=(c_1,0)$, then system~\eqref{c-def} is reduced to
\begin{eqnarray}
\label{c-def2}
\left\{\begin{array}{cll}
G(x+a_1,y+a_2)-c_1G(x,y)&=&b_1,\\
H(x+a_1,y+a_2)-c_1H(x,y)&=&b_2,
\end{array}\right.
\end{eqnarray}
which seems more hopeful to deal with. Thus, in what follows, we aim to construct more low $c$-differential uniformity functions $F(x,y)$, including the P$c$N and the AP$c$N functions with respect to $c\in\mathbb{F}_q\backslash\{1\}$.

\begin{prop}
\label{f-h1h2}
Let $F(x,y)=(g(x),H(x,y))$ with $H(x,y)=h_1(x)L_1(y)+\gamma_1h_2(x)+\gamma_2L_2(y)$, where $\gamma_1,\,\gamma_2\in\mathbb{F}_q$, $g(x),\,h_i(x)$ are functions from $\mathbb{F}_q$ to $\mathbb{F}_q$ and $L_i$ is a linearized polynomial over $\mathbb{F}_q$ for $1\leq i\leq 2$. Let $c=(c_1,c_2)\in\mathbb{F}_{q}^2\backslash\{(1,0)\}$ and $\gamma_2\ne0$.
 If $g(x)$ is a P$c$N function with respect to some $c\in\{(c_1,0):c_1\in\mathbb{F}_q\backslash\{1\}\}$ and   $L_2(y)+\gamma L_1(y)$ is either $2$-to-$1$ or permutation for any $\gamma\in\mathbb{F}_q$, then $_{c}\Delta_F\leq2.$
\end{prop}

\begin{proof}  Note that $c\in\mathbb{F}_q\backslash\{1\}$. Thus, one has $c=(c_1,0)$ for some $c_1\ne1$.  For any $a=(a_1,a_2)$, $b=(b_1,b_2)$, the $c$-differential equation of $F(x,y)$ can be expressed as
\begin{eqnarray*}
\left\{\begin{array}{cll}
g(x+a_1)-c_1g(x)&=&b_1,\\
\left(h_1\left(x+a_1\right)-c_1h_1(x)\right)L_1(y)+\gamma_2(1-c_1)L_2(y)&=&b'
\end{array}\right.
\end{eqnarray*}
by \eqref{c-def2}, where $$b'=b_2-\gamma_1\left(h_2(x+a_1)-c_1h_2(x)\right)-\gamma_2L_2(a_2)-h_1(x+a_1)L_1(a_2),$$ which is linear on $b_2$.
Since $g(x)$ is P$c$N and $L_2(y)+\gamma L_1(y)$ is 2-to-1 or permutation for any $\gamma\in\mathbb{F}_q$, then $g(x+a_1)-c_1g(x)=b_1$ has exactly one solution for any $a_1,\;b_1\in\mathbb{F}_q$ and there are at most two solutions for the second equation. Then the result follows.
\end{proof}

\begin{remark} Note that if $h_1$ is a non-zero constant polynomial in Theorem~\textup{\ref{f-h1h2}}, namely, $F$ has the form $F(x,y)=(g(x),h(y)+f(x))$, where $f(x)$ is any polynomial over $\mathbb{F}_q[x]$ and $h$ is not necessarily a linearized polynomial. In this case, we can check that if $g$ is $(c,\delta_1)$-uniform and $h$ is $(c,\delta_2)$-uniform for the same $c\in\mathbb{F}_q\backslash\{1\},$ then  $F$ is $(c,\delta_1\delta_2)$-uniform. Therefore, if we choose $g$ and $h$ such that both of them are either P$c$N or AP$c$N, then  the $c$-differential uniformity of $F$ is at most $4$. Based on the known P$c$N and AP$c$N functions, an abundance of new classes of P$c$N and AP$c$N functions can be produced by this way. For instance, let $g(x)=x^{\frac{p^{k_{1}}+1 }{2}} $ and $h(y)=y^{p^{k_2}+1}$ with $p>2$, $\frac{2m}{\gcd(2m,k_1)}$ and $\frac{m}{\gcd(m,k_2)}$ are odd, then $g$ is P$c$N and $h$ is AP$c$N for $c=-1$ by~\textup{\cite[Theorem 3 and Theorem 6]{SCP}}. By selecting $f$ at random, Magma experiments always show that $F(x,y)=(g(x),h(y)+f(x))$ is  AP$c$N  with respect to $c=-1$.
 \end{remark}

Next, we give certain examples to prove the existence of  functions in Theorem \ref{f-h1h2}.
\begin{example}
 Let $q=2^m$ and $g(x)$ be any linearized permutation polynomial over $\mathbb{F}_q$. Let $L_2(x)=x$ and $h_2(x)=x^{2^k}$ with $\gcd(k,m)=1$. Obviously, one can easily check that $L_2(x)+\gamma h_2(x)$ is a permutation when $\gamma =0$ or $2$-to-$1$ when $\gamma \ne0$.  Selecting $\gamma_1L_1(x)=x^2+x$ or $\gamma_1=0$, $h_1(x)=x^{-1}$ or $h_1(x)=x^{2^k+1}$, Magma always shows that $F(x,y)$ is AP$c$N for $c\in\mathbb{F}_q\backslash\{1\}$.
\end{example}

\begin{example} Let $q=3^m$ with $m$ odd and $F(x,y)=(x^{\frac{3^k+1}{2}},y^3+\gamma y)$, where $k$ is even and $\gamma$ is a square element in $\mathbb{F}_q$. Then one has $g(x)=x^{\frac{3^k+1}{2}}$ is P$c$N with respect to $c=-1$ from~\textup{\cite[Theorem 6]{SCP}} and $y^3+\gamma y$ permutes $\mathbb{F}_q$.  Magma experiments show that $F(x,y)$ is P$c$N for $c=(-1,0)$. This is consistent with the result in Theorem~\textup{\ref{f-h1h2}}.
\end{example}

By employing quadratic functions and linearized polynomials, we present two classes of functions as below.
\begin{prop}\label{x^(p^k+1)+y^(p^k+1)-L(x+y)} Let $q=p^m$ and $F(x,y)=(x^{p^k+1}+\gamma y^{p^k+1},L(x+y))$, where $L$ is a linearized permutation polynomial over $\mathbb{F}_q$ and $\gamma\in\mathbb{F}_q\backslash\{-1\}$. Let $c\in\{(c_1,0):c_1\in\mathbb{F}_q\backslash\{1\}\}$ and $d=\gcd(p^k+1,p^m-1)$. If $c,\,\gamma\in\mathbb{F}_{p^{(m,k)}}$, then $F$ is $(c,d)$-uniform. Otherwise, $F$ is $(c,p^{\gcd(m,k)}+1)$-uniform.
\end{prop}

\begin{proof}
According to \eqref{c-def2}, it suffices to determine the maximum number of solutions of
\begin{numcases}{}
(x+a_1)^{p^k+1}+\gamma(y+a_2)^{p^k+1}-c_1x^{p^k+1}-c_1\gamma y^{p^k+1}=b_1,\label{gold-1}\\
(1-c_1)L(x+y)+L(a_1+a_2)=b_2,\label{gold-2}
\end{numcases}
when  $a=(a_1,a_2)$, $b=(b_1,b_2)$ run over $\mathbb{F}_q^2$. From  \eqref{gold-2}, we have $y=L^{-1}\left(\frac{b_2-L(a_1+a_2)}{1-c_1}\right)-x$. Replacing it into
\eqref{gold-1} gives that
\begin{eqnarray}\label{G-c2=0}
A_1x^{p^k+1}+A_2x^{p^k}+A_3x+A_4=0,
\end{eqnarray} where
\begin{alignat*}{2}
A_1&=(1+\gamma)(1-c_1),&\hspace{20pt} A_2&=a_1-(b_2'+a_2-c_1b_2')\gamma,\\ A_3&=a_1^{p^k}-((b_2'+a_2)^{p^k}-c_1b_2'^{p^k})\gamma,&\hspace{20pt}  A_4&=a_1^{p^k+1}+((b_2'+a_2)^{p^k+1}-c_1 b_2'^{p^k+1})\gamma -b_1
\end{alignat*}
with $b_2'=L^{-1}\left(\frac{b_2-L(a_1+a_2)}{1-c_1}\right)$. Note that $A\ne0$ due to $c_1\ne 1$ and $\gamma\ne -1$. If $c,\,\gamma\in\mathbb{F}_{p^{(m,k)}}$, then $\left(\frac{A_2}{A_1}\right)^{p^k}=\frac{A_3}{A_1}$ and further, one has
$$\left(x+\frac{A_2}{A_1}\right)^{p^k+1}-\left(\frac{A_2}{A_1}\right)^{p^k+1}+\frac{A_4}{A_1}=0,$$
which is a shift of $x^{p^k+1}$. Thus, $F(x,y)$ is $(c,d)$-uniform in this case.

 If $c,\,\gamma\notin\mathbb{F}_{p^{(m,k)}}$, there always exist $A_1,\,A_2,\,A_3 $ such that $\left(\frac{A_2}{A_1}\right)^{p^k}\ne\frac{A_3}{A_1}$, since $a$ and $b$ are arbitrary. Let $x=B_1z+B_2$ with $B_1=\left(\frac{A_3}{A_1}-\left(\frac{A_2}{A_1}\right)^{p^k}\right)^{p^{m-k}}$ and $B_2=-\frac{A_2}{A_1}$, then we have $$z^{p^k+1}+z+\frac{A_1B_2^{p^k+1}+A_2B_2^{p^k}+A_3B_2+A_4}{A_1B_1^{p^k+1}}=0$$
which has $0,\,1,\,2$ or $p^{\gcd(k,m)}+1$ solutions by Lemma~\ref{mainlem}. When $m\ne 2k$, since  the constant term of the above equation is linear on $b_1$, the value $p^{\gcd(k,m)}+1$ is achievable by Lemma~\ref{mainlem}.  When $m=2k$, again by Lemma~\ref{mainlem}, we can see that the above equation has at most two solutions. However, returning to Equation~\eqref{G-c2=0}, we always have $(a_1,a_2)$, $(b_1,b_2)\in\mathbb{F}_q^2$ such that $A_2=A_3=0$, when $a_1=a_2=b_2=0$. Thus, Equation~\eqref{G-c2=0} has at most $p^k+1$ solutions when $m=2k$ in this case.

Note that $p^{\gcd(k,m)}+1\geq\gcd(p^k+1,p^m-1)$. Then the result follows by the above analysis, which completes the proof.
\end{proof}

\begin{remark}
Recall that $\gcd(2^k+1,2^m-1)=\frac{2^{\gcd(2k,m)}-1}{2^{\gcd(k,m)}-1}$ and $\operatorname{gcd}\left(p^{k}+1, p^{m}-1\right)=2$ if $p>2$ and $\frac{m}{\operatorname{gcd}(m, k)}$ is odd. Therefore, one can obtain that if $c,\,\gamma\in\mathbb{F}_{p^{(m,k)}}$ with $c\ne 1,\,\gamma\ne-1$ and $\frac{m}{\operatorname{gcd}(m, k)}$ is odd, then $F$ is P$c$N when $p=2$ and AP$c$N when $p>2$.
\end{remark}

\begin{example} Choosing $q=p^m=2^6$, $k=2$ and $L(x)=x$, then $d=\gcd(p^k+1,p^m-1)=1$. When $\gamma\in\mathbb{F}_{p^2}\backslash\{1\}$,  Magma shows that $F$ in Theorem~\textup{\ref{x^(p^k+1)+y^(p^k+1)-L(x+y)}} is P$c$N for any $c\in\mathbb{F}_{p^2}\backslash\{1\}$ and $F$ is $(c,5)$-uniform when $c\in\mathbb{F}_{q}\backslash\mathbb{F}_{p^2}$.
\end{example}

\begin{example} Choosing $q=p^m=3^3$, $k=2$ and $L(x)=x^3+ x$. For any $\gamma\in\mathbb{F}_{q}\backslash\{2\}$,  Magma shows $F$ is $(c,4)$-uniform when $c\in\mathbb{F}_{q}\backslash\{1\}$.
\end{example}

\begin{prop}
\label{xy-xy+L(x+y)}
Let $q=p^m$ and $F(x,y)=\left(xy,\sum_{i=1}^m \gamma_i(xy)^{p^i}+L(x+y)\right)$, where $L$ is a linearized permutation polynomial over $\mathbb{F}_q$ and $\gamma_i\in\mathbb{F}_q$ for $1\leq i\leq m$. If $\gamma_{l}\ne0$ only for $l\in\{i_1,\,i_2,\,\ldots,\,i_t\}$, where $1\leq t \leq m$. Let $d=\gcd(i_1,i_2,\ldots,i_t,m)$, then $F$ is AP$c$N for $c\in\mathbb{F}_{p^d}\backslash\{1\}$.
In particular, if $\gamma_i=0$ for all $1\leq i\leq m$, $F$ is AP$c$N for $c\in\{(c_1,0):c_1\in\mathbb{F}_q\backslash\{1\}\}$ and $F$ is always AP$c$N for $c=(0,0)$ regardless of the values of $\gamma_i'$s.
\end{prop}

\begin{proof} Let $c\in\{(c_1,0):c_1\in\mathbb{F}_q\backslash\{1\}\}$. By \eqref{c-def2}, the $c$-differential equation of $F(x,y)$ can be written as
\begin{numcases}{}
(1-c_1)xy+a_2x+a_1y+a_1a_2=b_1,\label{xy-1}\\
\sum_{i=1}^m\gamma_i\left(\left((x+a_1)(y+a_2)\right)^{p^i}-c_1(xy)^{p^i}\right)+(1-c_1)L(x+y)=b_2-L(a_1+a_2).\label{xy-2}
\end{numcases}
By taking $p^i$-th power on both sides of \eqref{xy-1} for $1\leq i\leq m$ in turn, \eqref{xy-2} can be reduced to
$$\sum_{i=1}^m\gamma_i\left(c^{p^i}-c\right)(xy)^{p^i}+(1-c)L(x+y)=b_2-L(a_1+a_2)-\sum_{i=1}^m b_1^{p^i}\gamma_i.$$
Let $\gamma_{i_j}\ne0$ for $1\leq j \leq t$ and $\gamma_{l}=0$ for $l\notin\{i_1,\,i_2,\,\ldots,\,i_t\}$. Since $d=\gcd(i_1,i_2,\ldots,i_t,m)$ and $c\in\mathbb{F}_{p^d}\backslash\{1\}$, the above equation can be reduced to
$$(1-c)L(x+y)=b_2-L(a_1+a_2)-\sum_{i=i_1}^{i_t}b_1^{p^i}\gamma_i.$$
Thus, $y$ can be uniquely determined by $x$, that is, $$y=L^{-1}\left(\frac{b_2-L(a_1+a_2)-\sum_{i=i_1}^{i_t}b_1^{p^i}\gamma_i}{1-c}\right)-x.$$
Substituting it into Equation~\eqref{xy-1},  then we can obtain a quadratic equation which has at most two solutions in $\mathbb{F}_q$ and it can  attain two solutions when choosing $a,\,b$, properly. For example, if $a_1=a_2=b_1=0$ and $b_2\ne0$, then Equation~\eqref{xy-1} has the form $x^2+tx=0$ with $t\ne0$, which has exactly two solutions for any $c\in\mathbb{F}_q$.

In particular, if $\gamma_i=0$ for all $1\leq i\leq m$ or $c=0$, it is immediate that $F$ is AP$c$N by the above analysis. This completes the proof.
\end{proof}

\begin{example}
Let $q=2^m$ and $F(x,y)=\left(xy, (xy)^{2^4}+(xy)^{2^2}+x+y\right)$. Magma experiments show that $F$ is an AP$c$N function for $c\in\mathbb{F}_{2^2}$ and $c\ne1$.
\end{example}

%
%
%Let  $F(x)$ be a function from $\mathbb{F}_{q^2}$ to $\mathbb{F}_q\times\mathbb{F}_q$, by a similar process as before, we can also derive the expression of the $c$-differential equation of $F(x)$.  Let  $n=2m$. In the following, we present two classes of functions $F(x)$ which have low $c$-differential uniformity by using the Gold function and the inverse function.
As constructed in  \eqref{inv},  if  $F$ is a function from $\mathbb{F}_{q^2}$ to $\mathbb{F}_q\times\mathbb{F}_q$, we can give the following results  (the first one can be similarly proved as in Theorem~\ref{inv}).

\begin{prop}
\label{inv1}
Let $H$ be a function from $\mathbb{F}_{q^2}$ to $\mathbb{F}_{q}$ and $F(x)=\left(\Tr^n_m(x),\, H(x)\right)$. Let $c=(c_1,c_2)\in\mathbb{F}_q^2\backslash\{(1,0)\}$.
\begin{enumerate}
\item[\rm{(1)}]  When $H(x)=\Tr^n_m\left(\gamma x^{p^k+1}\right)$, where $\gamma^q+\gamma\ne 0$.  If $\gcd(k,m)=1$, then $_{c}\Delta_F(x)\leq6$;  if $c\in\{(c_1,0):c_1\in\mathbb{F}_q\backslash\{1\}\}$, then
$F(x)$ is $(c,p^{\gcd(k,m)}+1)$-uniform;
\item[\rm{(2)}] When  $H(x)=x^{q+1}$,  then
$F(x)$ is AP$c$N for  $c\in\{(c_1,0):c_1\in\mathbb{F}_q\backslash\{1\}\}$.
\end{enumerate}
\end{prop}

\begin{remark}
In Proposition~\textup{\ref{x^(p^k+1)+y^(p^k+1)-L(x+y)}}, $F(x,y)$ has the univariate form $F_1(z)=\frac{1}{\left(\beta-\overline{\beta}\right)^{p^k+1}}((\beta^{p^k+1}+\gamma)\overline{z}^{p^k+1}+(\overline{\beta}^{p^k+1}+\gamma)z^{p^k+1}-(\overline{\beta}^{p^k}\beta
+\gamma)\overline{z}z^{p^k}-(\overline{\beta}\beta^{p^k}
+\gamma)\overline{z}^{p^k}z)+\beta L(\frac{(\overline{\beta}-1)z-(\beta -1)\overline{z}}{\overline{\beta}-\beta}) $,  while $F$ from Proposition~\textup{\ref{inv1}(1)} has the univariate form  $F_2(z)=\beta\gamma z^{p^k+1}+\beta\overline{\gamma} \overline{z}^{p^k+1}+z+\overline{z}$ by \eqref{univariate}. We can see that $F_1(z)$ is not equal to $F_2(z)$. Recall that   the $c$-differential uniformity of a given function $F(x)$ is preserved through $F\circ L$ (note that it is not preserved through $L_1\circ F\circ L_2$ for affine permutations $L_1$ and $L_2$) and is not invariant under EA-equivalence and CCZ-equivalence mentioned in~\textup{\cite{HPRS}}. Therefore, we claim that  $F$ of Proposition~\textup{\ref{inv1}(1)} is not equivalent to the function
$F$ of Proposition~\textup{\ref{x^(p^k+1)+y^(p^k+1)-L(x+y)}}. So do the functions $F$ of Proposition~\textup{\ref{inv1}(2)} and of Proposition~\textup{\ref{xy-xy+L(x+y)}}.
\end{remark}

\begin{prop}\label{x^(q+1)}
Let $H$ be a function from $\mathbb{F}_{q^2}$ to $\mathbb{F}_{q}$ and $F(x)=\left(x^{q+1}, H(x)\right)$. Let $c=(c_1,0)\in\mathbb{F}_q\backslash\{1\}$ and $U=\{x\in\mathbb{F}_{q^2}:x^{q+1}=1\}$. Then $F$ is $(c,\delta)$-uniform if $H(x+a)-c_1H(x)=b_2$ has at most $\delta$ solutions on $\beta U+\frac{a}{c_1-1}$ for any $\beta\in\mathbb{F}_{q^2}$ and $b_2\in\mathbb{F}_q$.
\end{prop}

\begin{proof}  To completes the proof, we need to prove $F(x+a)-cF(x)=b$ has at most $\delta$ solutions for any $a\in\mathbb{F}_{q^2}$ and $b=(b_1,b_2)\in\mathbb{F}_q^2$. By \eqref{c-def2}, one has
$(x+a)^{q+1}-c_1x^{q+1}=b_1.$
Let $x=y+\frac{a}{c_1-1},$  then $y^{q+1}=\frac{b_1(c_1-1)-ca^{q+1}}{(c_1-1)^2}$, which always has solution with the form $y=\beta U$, $\beta\in\mathbb{F}_{q^2}$ and $\beta^{q+1}=\frac{b_1(c_1-1)-ca^{q+1}}{(c_1-1)^2}$. This completes the proof.
\end{proof}
%\section*{Acknowledgements}
\begin{example}
 Let $q=2^4$ and $H(x)=\Tr^n_m(x^5)$ and $c=(c_1,0)$. Magma experiments show that $F$ in Theorem~\textup{\ref{x^(q+1)}} is AP$c$N if  $c_1\in\mathbb{F}_{2^2}\backslash\{1\}$ and otherwise, $F$ is $(c,6)$-uniform.
\end{example}
\section{Conclusions}
In this paper, we mainly focused on the construction of bivariate functions $F(x,y)\in\mathbb{F}_q[x,y]$ with low $c$-differential uniformities. By analyzing the relationship between bivariate functions and univariate functions, we gave a new concept of the $c$-differential equation for bivariate functions. By virtue of some known functions, such as the Gold function, the inverse function, the trace function and linearized polynomials, we   proposed four classes of bivariate functions with low $c$-differential uniformity for any $c\in\mathbb{F}_{q}^2\backslash\{(1,0)\}$ while several classes of bivariate functions with low $c$-differential uniformity for  $c\in\mathbb{F}_{q}\backslash\{1\}\times\{0\}$.  In particular, P$c$N and AP$c$N functions could be found from our constructions. In particular, this adds to the very few known classes of P$c$N functions in even characteristic (there are only two non-trivial classes of such, besides sporadic examples).

\end{document}